\documentclass[12pt,a4paper]{article}
\usepackage{amssymb}
\usepackage{amsfonts}
\usepackage[english]{babel}
\usepackage{graphicx}
\usepackage{appendix}
\usepackage{amsmath}
\usepackage{amsthm}

\newtheorem{theorem}{Theorem}[section]
\newtheorem{proof*}{Proof }

\newtheorem{definition}[theorem]{Definition}

\newcommand{\bean}{\begin{eqnarray*}}
\newcommand{\eean}{\end{eqnarray*}}
\newcommand{\ba}{\begin{array}}
\newcommand{\ea}{\end{array}}
\newcommand{\be}{\begin{equation}}
\newcommand{\ee}{\end{equation}}
\newcommand{\bea}{\begin{eqnarray}}
\newcommand{\eea}{\end{eqnarray}}
\newcommand{\pa}{\partial}

\newcommand{\no}{\nonumber}

\newcommand{\bk}{\backslash}

\newcommand{\qd}{\quad}

\begin{document}

\title
{Totally Non-negative Pfaffian for Solitons in BKP Equation}
\author{
 Jen-Hsu Chang \\Graduate School of National Defense, \\
 National Defense University,\\
Tau-Yuan City, 335009, Taiwan}

\date{}

\maketitle
\begin{abstract}
The BKP equation is obtained from the reduction of B-type in the KP
hierarchy under the orthogonal type transformation group for the KP equation. The skew Schur’s Q functions can be used to construct the $\tau$-functions of solitons in the BKP equation. Then the totally non-negative Pfaffian can be defined via the skew Schur’s Q functions to obtain non-singular line-solitons  solution in the BKP equation. The totally non-negative Pfaffians are investigated. The line solitons interact to form web-like structure in the near field region and their resonances appearing in soliton graph could be investigated
by the totally non-negative Pfaffians.
\end{abstract}

\section{Introduction} 
\indent The BKP equation \cite{da, dj} or the 2+1 Sawada-Kotera equation \cite{ji} 
\be 
(9 \phi_t-5\phi _{xxy}+\phi_{xxxxx}-15 \phi_x \phi_y+15 \phi_{x}\phi_{xxx}+15\phi_x^3)_x-5 \phi_{yy}=0 \label{bkp}
\ee
is obtained from the reduction of B-type in the KP hierarchy under the orthogonal type transformation group for the KP equation.  It can also be obtained by the Kupershmidt reduction \cite{ku} or the hierarchy defined
on integrable 2D Schrodinger operators \cite{kr}. The web-solitons solutions of the BKP equation (\ref{bkp})  are constructed by the vertex operators and the Clifford algebra of free fermions \cite{dj} ($\tau$-function theory). The Hirota bilinear form  \cite{hi, ma, rh, yq} of BKP equation is obtained from the Clifford group acting on the vacuum.  The  Pfaffian structure for web solitons is a special solution ansatz for the Hirota bilinear form.  In \cite{kv, nim, ni, or}, the rational solutions are established using the $\tau$-functions expressed as  the linear combination of the Schur Q-functions (polynomials) over the Pfaffian coefficients defined on  partitions with distinct parts. On the other hand, the $\tau$-function of the BKP hierarchy is known to  be obtained as the partition function or as the matrix integrals \cite{hu, os} so that  one  could  study the Pfaffian point process \cite{wa}. The non-commutative case of BKP equation and its $\tau$-function are  investigated in \cite{de}. 

On the other hand, the resonant interaction plays a fundamental  role in multi-dimensional wave phenomenon. The resonances of web solitons of KP-(II) equation
\be    \pa_x (-4 u_t+u_{xxx}+6uu_x)+ 3u_{yy}=0  \label{kp} \ee
has attracted much attractions using the totally non-negative Grassmannians \cite{bc, ko1, ko3}.  For the KP-(II) equation case, the  $\tau$-function  is described by the Wronskian form obtained from the Hirota bilinear form \cite{hi}. Using the Cauchy-Bitnet formula, the Wroskian can be expressed as the linear combination of Plucker coordinates with dominant phase functions. To obtain  non-singular multi-line solitons solutions, the Plucker coordinates must be totally non-negative. A similar consideration also yields the totally non-negative Pfaffian since the $\tau$-function of BKP is also a linear combination of coefficients of Pfaffian with dominant phase functions. The resonance of $\tau$-function of Pfaffian structure is investigated in DKP theory \cite{km} using the A-soliton and D-soliton. Inspired by their results, one would study the resonance theory in BKP equation via the totally non-negative Pfaffian. 

The paper is organized as follows: In Section 2,  one derives the Hirota equation using the Clifford algebra as well as studies  the $\tau$-functions on skew-Schur's Q functions. In Section 3, one investigates the totally non-negative Pfaffians in block form. In Section 4, one investigates the singular totally non-negative Pfaffians in  theoretical framework  especially useful while dealing with the totally non-negative Grassmannian. Finally, one finishes the article with some conclusive remarks in Section 5.
   
\section{Solitons in Totally Non-negative Pfaffian } 

In this section, one  considers the $\tau$-function over Schur's Q-function  \cite{ni, or} in the BKP equation (\ref{bkp}) and generalizes it to  the skew Schur's Q-function. To obtain non-singular web solitons, one also introduces  the totally non-negative Pfaffian
defined in \cite{yt}. \\
\indent We start with Clifford algebra for BKP equation\cite{da, dj, lu} and  consider neutral fermions $ \{ \phi_n, n \in Z\}$, obeying the following canonical anti-communication relations 
\be 
 [\phi_m, \phi_n]_+=\phi_m \phi_n + \phi_n  \phi_m=(-)^m \delta_{m, -n}. \label{at} 
\ee
In particular, $ \phi_0^2=1/2$. \\
\indent There are the right and the left vacuum vectors $|0>$ and $<0|$ respectively, having the properties
\[ \phi_m|0>=0, (m<0), \qquad <0|\phi_m  =0 , ( m> 0),\]
and 
\[ \sqrt{2} \phi_0 =| 1 >, \qd  \sqrt{2} \phi_0 | 1 >= 0 | >,  \qd < 0 | \sqrt{2} \phi_0= < 1 |, \qd   < 1 | \sqrt{2} \phi_0=< 0 | . \]
We have a right and left Fock spaces spanned,  respectively, by the right and the left vacuum vectors 
\bea 
&& \phi_{n_1} \phi_{n_2} \phi_{n_3} \cdots \phi_{n_k} |0> , \no \\
&&<0|  \phi_{-n_1} \phi_{-n_2} \phi_{-n_3} \cdots \phi_{-n_k}, \label{fo}
\eea
where $k=1, 2, 3, \cdots $ and we assume 
\[ n_1> n_2>n_3 \cdots > n_k \geq 0 \]
due to the anti-communication relations $(\ref {at})$. The vacuum expectation values of quadratic elements are
given by
\[ < 0| \phi_i \phi_j | 0>=<  \phi_i \phi_j >=\left \{  \begin{array}{ll} 
(-1)^i \delta_{i, -j},  i <0,  \\
\frac{1}{2} \delta_{j, 0} , i=0, \\
0, i>0 .
\end{array}  \right. \]
Then one introduces the following Hamiltonian operators
labeled by odd numbers $ n \in 2Z+1$:
\be H_n=\frac{1}{2} \sum_{k= - \infty}^{\infty} (-1)^{k+1} \phi_k \phi_{-k-n}. \label{ha}
\ee
One can check that $H_n$ obey the following Heisenberg algebra relations
\be 
[ H_n, H_m]=H_n H_m-H_m H_n= \frac{n}{2} \delta_{n+m, 0}.
\ee
The anti-communication relations $(\ref {at})$ imply 
\be 
[ H_n, \phi_m]= \phi_{m-n}. \label{mi}
\ee
We also notice that by the definition of vacuum 
\[H_n |0 >=H_n \phi_0 |0 >=0. \]
For the independent time variables of BKP equation,  one sets 
\[ H(t_1, t_3, t_5, \cdots,   )= \sum_{n= 1, 3, 5, \cdots }^{\infty} t_n H_n .\]
It's suitable to introduce the following fermionic field, which depends on a complex parameter $p$ :
\[ \phi (p) = \sum_{k= -\infty }^{\infty} p^k \phi_k.\]
For fermionic fields, the anti-communication relations $(\ref {at})$ reads as 
\[ [\phi (p), \phi (p')]_+= \sum_{k= -\infty }^{\infty} (\frac{-p}{p'})^k = \delta ( \frac{-p}{p'}) ,\]
and when $ |p| \neq |p'|$, one has 
\[ < 0| \phi (p) \phi (p') |0>=< \phi (p) \phi (p')>= \frac{(p-p')}{2(p+p')}. \] 
By the Wick theorem, one gets
\bea
< \phi (p_1) \phi (p_2) \phi (p_3)  \cdots  \phi (p_N) > &=& \left \{  \begin{array}{ll} 
Pf (<\phi (p_i) \phi (p_j)>,  \mbox {N even}\\
0, \mbox {N odd } 
\end{array}  \right.   \no \\
&=& \left \{  \begin{array}{ll} 
2^{-N/2} \prod_{i<j} \frac{p_i-p_j}{p_i+p_j},  \mbox {N even}\\
0, \mbox {N odd } 
\end{array}  \right.  \label{wi}
\eea
where $Pf $  means the Pfaffian. 
The relation $(\ref {mi})$ yields 
\[ [ H_n, \phi (p)]= z^n \phi (p),\]
which in turn results in
\[ \phi(p)(t)=e^{ H(t) } \phi(p)e^{ -H(t) }= \phi(p)e^{ \xi (t,p) }, \quad \xi (p,t)= \sum_{k=1, 3, 5 , \cdots }^{\infty} t_k p^k. \]
To introduce the $\tau$-function, we define the Clifford group 
\be 
G= \{g | g^{-1} \quad \mbox {exists and } g \phi_k g^{-1}=\sum_{j \in Z } a_{jk} \phi_j \}. \label{gu}
\ee
A typical element of $G$ is the exponential 
\[ g=e^{\frac{1}{2} \sum_{r , s \geq 0 } \alpha_{rs} \phi_r \phi_s }\]
of a fermion bilinear form. Such an fermion operator defines a $\tau$-function of the BKP equation 
\be 
\tau_g (t)= <0|  e^{ H(t) } g |0>. \label{ta}
\ee
The space of the BKP $\tau$-function is the G-orbit of the vacuum vector  $|0>$.  The Hirota equation is 
\be 
\oint \frac{dz}{2\pi i z} e^{\xi (t'-t, z)} \tau (t'- 2[\frac{1}{z}] ) \tau (t+ 2[\frac{1}{z}] ) dz= \tau (t) \tau (t'), \label{hir}
\ee
where 
\[ [\frac{1}{z}]=( \frac{1}{z} ,  \frac{1}{3z^3},  \frac{1}{5z^5}, \cdots]. \]
The Hirota equation (\ref{hir}) is the consequence of the algebraic relation
\[ \sum_{j \in Z }(-1)^j \phi_j g \otimes \phi_{-j} g=  \sum_{j \in Z }(-1)^j g \phi_j  \otimes g \phi_{-j} \]
and the formula called the Boson-Fermion Correspondence \cite{ak}: 
\be
\sqrt{2}<1| e^{H(t)} \phi (p) | V>= X(t, z) <0| e^{H(t)} | V>,   \label{bf}
\ee 
holding  for arbitrary Fock space V defined in (\ref{fo}). Here  $X(t,z)$ is the vertex operator
\[  X(t,z)=  e^{\xi ( t,z)  } e^{-2 D(t, 1/z)}, \quad  D(t, 1/z)=  \sum_{n=0 }^{\infty} \frac{z^{-2n-1}}{2n+1}  \pa_{2n+1} .\]

\indent Next, the Schur Q polynomials are defined by , $c=(c_1, c_3, c_5, c_7, \cdots ) $ 
\[ e^{2 \xi (c,p)}=\sum_{k=0} q_k (c ) p^k.\]
For example, 
\bean 
&& q_0(c)=1, \quad q_1(c)= 2c_1, \quad q_2(c)=(4c_1^3)/3 + 2c_3, \quad q_4(c)= 4c_1c_3 + 2/3c_1^4, \\
&&  q_5(c)=2c_5 + 4c_1^2c_3 + 4/15c_1^5, \quad q_6(c)=4c_1c_5 + 8/3c_3c_1^3 + 2c_3^2 + 4/45c_1^6, \\
&& q_7(c)= 2c_7 + 4c_1^2c_5 + 4c_1c_3^2 + 4/3c_3c_1^4 + 8/315c_1^7, \\
&& q_8(c)=4c_1c_7 + 8/3c_5c_1^3 + 4c_5c_3 + 4c_1^2c_3^2 + 8/15c_3c_1^5 + 2/315c_1^8.
\eean
Thus 
\[\phi_i(c )=e^{H(c )} \phi_i | 0>=\sum_{k=0} q_k (c /2) \phi_{i-k}. \]
We have  \cite{ni}
\be <\phi_i(c ) \phi_j(c )>=\frac{1}{2} q_i (c /2) q_j (c /2)+ \sum^j_{k=1}(-1)^k q_{i+k } (c /2) q_{j-k} (c /2).\label{phi}
\ee
Since,  $ q_k (c /2)=0 $  if  $k<0$, 
\[ 1= e^{ 2\xi(p,c )} e^{ -2\xi(p,c )}= \sum_{i, j} q_i (c ) q_{j-i}  (-c )=  \sum_{i, j} (-1)^{i-j} q_i (c ) q_{j-i}  (c )p^j, \]
we have the orthogonal condition for all $n>0$ 
\be   
\sum^n_{i=0} (-1)^{i} q_i (c ) q_{n-i}  (c )=0. \label{or}
\ee
This is trivial if $n$ is odd and if $n=2m$ is even, then it gives 
\[ q_m (c )^2+2 \sum^m_{k=1} (-1)^{k} q_{m+k} (c ) q_{m-k}  (c )=0.\]
Then one can define 
\be 
q_{a ,b}(c )= q_a (c )q_b (c )+2 \sum^b_{k=1} (-1)^{k} q_{a+k} (c) q_{b-k}  (c ). \label{qf}
\ee
Here we notice that $q_{a,0}(c)= q_a(c)$.
It follows from the orthogonal condition (\ref{or}) that 
\[ q_{a,b}(c )=-q_{b,a}(c ),\]
and in particular, $ q_{a,a}(c )=0$. Comparing (\ref{phi}) and (\ref{qf}), one has 
\be 
q_{a,b}(c/2)=2  <\phi_i(c ) \phi_j(c )>.  \label{pf}
\ee
Now, consider $\lambda=(\lambda_1, \lambda_2,  \cdots, \lambda_{2n})$ , where $ \lambda_1> \lambda_2 > \lambda_3> \cdots > \lambda_{2n} \geq 0. $ The set of such distinct partition is denoted DP. For $\lambda \in DP$, 
we define
\be 
Q_{\lambda}(c /2)= Pf( q_{\lambda_i, \lambda_j}(c /2)). \label{ch}
\ee
This is the Schur Q function \cite{mac}. Here we notice that $Q_{\lambda_i, \lambda_j}=q_{\lambda_i, \lambda_j}$ and $ Q_{\lambda,0}= Q_{\lambda}=q_{\lambda}$. By the Wick Theorem, 
\bean
Q_{\lambda}(c /2)& =& Pf( 2 <\phi_{\lambda_i} (c ) \phi_{\lambda_j} (c )>) \\
&=& 2^n < \phi_{\lambda_1} (a)\phi_{\lambda_2} (a) \cdots \phi_{\lambda_{2n}} (a) > .
\eean 
Then the $\tau$-function \cite{or} of soliton solution of BKP is , noting that  $\xi (p , t)=e^{p x+p^3 y+ p^5 t}$, 
\bea 
\tau(c, t) &= & < e^{H(t)} e^{ \sum_{ 0 \leq i<j} q_{i,j}(c/2) \phi (p_i)  \phi_(p_j)} > \no \\  
&=& 1+ \sum_{0 \leq i<j} Q_{i,j}(c /2) \frac{1}{2} \frac{p_i-p_j}{p_i+p_j} e^{\xi (p_i, t)+\xi (p_j , t)}   \no \\
&+ & \sum_{0 \leq i<j<k<l} Q_{i,j,k, l}(c /2) \frac{1}{2^2}  \frac{p_j-p_i}{p_j+p_i} \frac{p_k-p_i}{p_k+p_i} \frac{p_l-p_i}{p_l+p_i} \frac{p_k-p_j}{p_k+p_j}\frac{p_l-p_j}{p_l+p_j}\frac{p_l-p_k}{p_l+p_k} \no \\
&& e^{\xi (p_i, t)+\xi (p_j , t)+ \xi (p_k, t)+ \xi (p_l, t)} \no \\
&+&\sum_{0 \leq i<j<k<l<m<n } Q_{i,j,k, l, m, n}(c /2) \frac{1}{2^3}\prod_{0 \leq  \alpha < \beta \leq n} \frac{p_{\alpha}-p_{\beta}}{p_{\alpha}+ p_{\beta}} e^{\xi (p_i, t)+\xi (p_j , t)+ \xi (p_k, t)+ \xi (p_l, t)+\xi (p_m, t)+\xi (p_n, t)} \no \\
&+& \cdots  \label{tc}
\eea
Given  a  Young diagram of DP   $ \lambda=(\lambda_1,  \lambda_2 ,  \lambda_3,  \cdots ,  \lambda_{2n} \geq 0) $, 
there are 
\[ {2n \choose 0} + {2n \choose 2 } + {2n \choose 4} + \cdots + {2n \choose 2n}  =2^{2n-1}\]
terms in the expansion (\ref{tc}). The sum in (\ref{tc}) is over all the DP of  $ \lambda$ and this corresponds to a $2n-1$ soliton. One remarks here that if $Q$ is any anti-symmetric matrix, then (\ref{tc}) is also a $\tau$-function for the BKP equation \cite{ni}. 

We notice the $\tau$-function (\ref{tc}) has the form 
\be 
\tau=\sum_{I \subset [1, 2, \cdots , 2n] } e^{ \xi (p_{i_1} , t)+\xi (p_{i_2} , t)+ \xi (p_{i_3}, t)+ \cdots +\xi (p_{i_{2m}}, t)+  \ln \rho_I} , \label{roo} 
\ee
where $ I=\{ i_1, i_2, \cdots , i_{2m} \} $ and 
\[ \rho_I=Q_I (\frac{c}{2}) \frac{1}{ 2^{m}} \prod_{0 \leq  \alpha < \beta \leq i_{2m}} \frac{p_{\alpha}-p_{\beta}}{p_{\alpha}+ p_{\beta}} .\]
The soliton structure is determined from the consideration of the dominant exponentials $ e^{ \xi (p_{i_1} , t)+\xi (p_{i_2} , t)+ \xi (p_{i_3}, t)+ \cdots +\xi (p_{i_{2m}}, t)}$ in the $\tau$-function at different regions of the $(x,y)$-plane. The soliton solution $u=2 \pa_x^2 \ln \tau$ is localized at the boundaries of two distinct regions where a balance exists between two dominant exponentials in (\ref{roo}). In each interior of these regions, the solution is exponentially small except the dominant exponential of a specific index set. We take two index sets $I=\{ i_1, i_2, \cdots , i_{2m} \} $  and $J=\{ j_1, j_2, \cdots , j_{2s} \} $. Then
\[ e^{\Theta_I }+ e^{\Theta_J}= 2 e^{\frac{\Theta_I +\Theta_J} {2}} \cosh \frac{\Theta_I -\Theta_J} {2}, \]
where 
\bean 
 \Theta_I &= & \xi (p_{i_1} , t)+\xi (p_{i_2} , t)+ \xi (p_{i_3}, t)+ \cdots +\xi (p_{i_{2m}}, t)+  \ln \rho_I \\
 \Theta_J &=& \xi (p_{j_1} , t)+\xi (p_{j_2} , t)+ \xi (p_{j_3}, t)+ \cdots +\xi (p_{j_{2s}}, t)+  \ln \rho_J.
\eean 
Close to the boundaries of  two dominant exponential $ \Theta_I $ and $ \Theta_J $, one yields 
\bea
2 \pa_x^2 \ln  \tau & \approx &  2 \pa_x^2 \ln ( e^{\Theta_I }+ e^{\Theta_J})= \frac{(\sum_{\alpha=1}^{2m}p_{i \alpha} -  \sum_{\beta=1}^{2s}p_{j \beta})^2 } {2} {sech^2} \frac{\Theta_I - \Theta_J}{2} \\
&=& \frac{(\sum_{\alpha=1}^{2m}p_{i \alpha} -  \sum_{\beta=1}^{2s}p_{j \beta})^2 } {2} {sech^2} \frac{ \textit{K}_{[I,J]} \centerdot  (x,y)+ \Omega_{[I,J]} t + \ln \frac{\rho_I}{\rho_J}}{2}, \label{am}
\eea   
where $    \textit{K}_{[I,J]} $ is the wave vector and   $\Omega_{[I,J]}   $  is the frequency defined by 
\bean
 \textit{K}_{[I,J]}& =& ( \sum_{\alpha=1}^{2m}p_{i \alpha} -  \sum_{\beta=1}^{2s}p_{j \beta}, \qd  \sum_{\alpha=1}^{2m}p_{i \alpha}^2-   \sum_{\beta=1}^{2s}p_{j \beta}^2  )=(\textit{K}_{[I,J]}^x,   \textit{K}_{[I,J]}^y ), \\
\Omega_{[I,J]} &=& \sum_{\alpha=1}^{2m}p_{i \alpha}^3-   \sum_{\beta=1}^{2s}p_{j \beta}^3. 
\eean  
The direction of wave vector $  \textit{K}_{[I,J]}$ is measured in the counter-clockwise sense from the $y$-axis, and it is given by 
\[ \tan \Phi_{[I,J]}= \frac{     \textit{K}_{[I,J]}^y    }{ \textit{K}_{[I,J]}^x}    =\frac{ (\sum_{\alpha=1}^{2m}p_{i \alpha}^2)- (  \sum_{\beta=1}^{2s}p_{j \beta}^2 )   }{ \sum_{\alpha=1}^{2m}p_{i \alpha} -  \sum_{\beta=1}^{2s}p_{j \beta}                         }\]      
This soliton is localized along $\Theta_I - \Theta_J=0$ and  has the amplitude $\frac{(\sum_{\alpha=1}^{2m}p_{i \alpha} -  \sum_{\beta=1}^{2s}p_{j \beta})^2 } {2} $ with the phase shift $ \ln \frac{\rho_I}{\rho_J}$. The soliton velocity is given by 
\[ V_{[I,J]}= \frac{-\Omega_{[I,J]}} {|\textit{K}_{[I,J]}|^2} \textit{K}_{[I,J]}=\frac{-\Omega_{[I,J]}} {(\textit{K}_{[I,J]}^x)^2+ (\textit{K}_{[I,J]}^y)^2} ( \textit{K}_{[I,J]}^x, \textit{K}_{[I,J]}^y).\]
The resonance condition among those three line-solitons is given by 
\be    \textit{K}_{[I,J]}=    \textit{K}_{[I,L]}  +\textit{K}_{[L,J]}, \qd  V_{[I,J]}=V_{[I,L]}    + V_{[L,J]}. \label{ro}   \ee
If some of the in-dices are identical, that is, 
\[i_1=j_1, i_2=j_2, \cdots , i_{2k}=j_{2k}, \qd k< min \{ m,s \}, \]
then it is called the $ [I \backslash K,   J\backslash K] $-soliton, where $K=\{i_1=j_1, i_2=j_2, \cdots , i_{2k}=j_{2k} \}. $ One remarks here that in  the KP solitons theory of resonance \cite{bc, ko6} the soliton is the balance of two dominant exponentials of the index sets $I =\{i, i_2, i_3, \cdots,    i_{m} \},  J  =\{j , i_2, i_3,  \cdots,   i_{m} \} $ and it is called the $[i,j]$-soliton.\\

Take 
\[  Q= \left [\ba {cccc} 0 &Q_{12}  &Q_{13} & Q_{14} \\
    -Q_{12} & 0&Q_{23} &Q_{24} \\   -Q_{13} &-Q_{23}  &0 & Q_{34}   \\ -Q_{14} &  -Q_{24}&-Q_{34}& 0
\ea \right]. \]
Then
\bea 
\tau_Q  &=& 1+ \sum_{ i<j} Q_{i,j} \frac{1}{2} \frac{p_i-p_j}{p_i+p_j} e^{\xi (p_i, t)+\xi (p_j , t)}   \no \\
&+ & Q_{1234} \frac{1}{2^2}  \frac{p_1-p_2}{p_1+p_2} \frac{p_1-p_3}{p_1+p_3} \frac{p_1-p_4}{p_1+p_4} \frac{p_2-p_3}{p_2+p_3}\frac{p_2-p_4}{p_2+p_4}\frac{p_3-p_4}{p_3+p_4} \no \\
&& e^{\xi (p_1, t)+\xi (p_2 , t)+ \xi (p_3, t)+ \xi (p_4, t)}, \no \label{tq}
\eea
where $Q_{1234}=Pf(Q)=Q_{12}Q_{34}-Q_{13}Q_{24}+Q_{14}Q_{23}.$ Please see the figure 1. \\
\begin{figure}[t]
	\centering
		\includegraphics[width=1\textwidth]{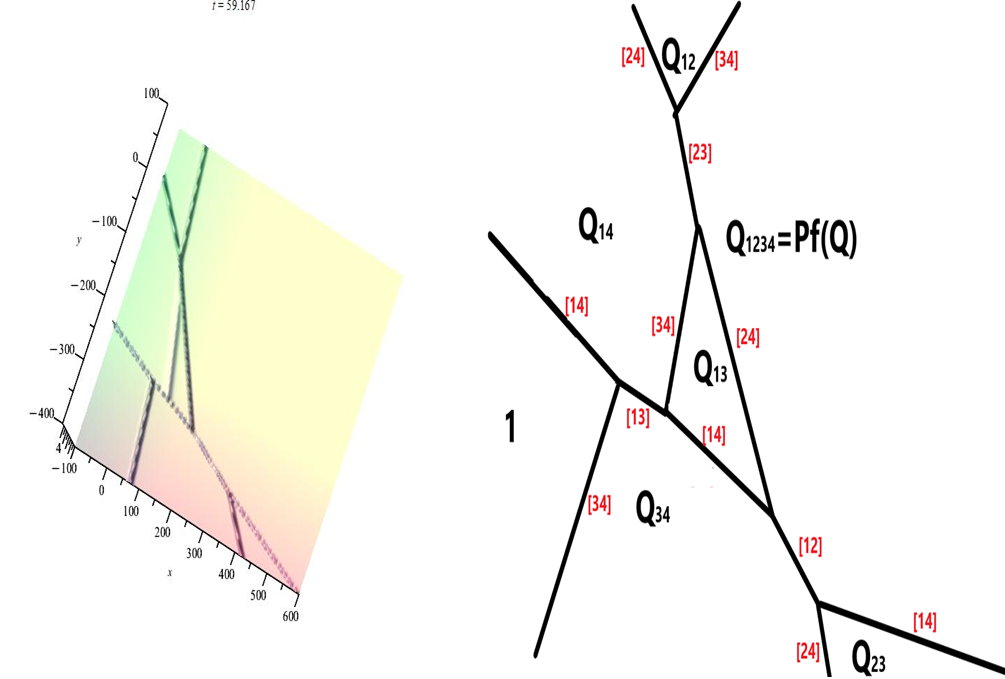}
	\caption{$p_1=2, p_2=1, p_3=0.5, p_4=0.2, Q_{12}=1, Q_{13}=2, Q_{14}=3, Q_{23}=4, Q_{24}=5, Q_{34}=6.$  The left panel is the soliton graph of  the $\tau$-function (\ref{tq}). The right panel shows that the plane is divided several regions, each having their own  corresponding dominant exponential (black color). The soliton $[IJ]$ (red color) is localized at the boundaries of two distinct regions. The three line-solitons $[ik], [kj] $ and $ [ij] $ form a resonant Y-type (or upside down) soliton.  There exists a triangle in the middle due to the resonance.}
\end{figure}
\indent To obtain non-singular resonant-solitons solutions of the BKP equation (\ref{bkp}), one has to assume $Q_{DP \in \lambda}(c) \geq 0$ for all DP of $ \lambda$ and the conditions
\[ p_1^2  > p_2^2 > p_3^2> \cdots > p_{2n}^2 \geq 0. \] 
Then one has the following
\begin{definition}
Let $A$ be a $2n \times 2n$ skew-symmetric matrix. Then $A$ is a totally positive (non-negative) Pfaffian if every $2m \times 2m ( m \leq n ) $ principal sub skew-symmetric matrix of $A$ has a positive (non-negative)  Pfaffian. 
\end{definition}

Fortunately, an alternative purely combinatorial definition has been given in terms of shifted Young diagrams \cite{st} and  $Q_{\lambda}(c) $ is the generating function summed over all marked shifted tableaux of shape $\lambda$. It can be described as follows. Given  a strict Young diagram   $ \lambda=(\lambda_1,  \lambda_2 ,  \lambda_3,  \cdots ,  \lambda_{2n} \geq 0) $ with length $2n$, a  shifted Young diagram is defined via 
\be 
D=\{ (i, j) \in \textsl{Z}^2 |  i \leq j \leq \lambda_i +i-1,  0 \leq j \leq 2n\}. \label{id}
\ee
Let $(1,2,3, \cdots, m) $ and $(1', 2', 3', \cdots , m')$  be two sequences of symbols ordered in such a way $ 1'<1<2'<2<3'<3 \cdots < m'<m$. A  tableau $T$ of a shifted Young diagram is an assignment 
\be
T : D  \to (1,2,3, \cdots, m, 1',2', 3', \cdots , m')    \label{tab}
\ee
such that \\
(1)$ T(i,j) \leq T(i+1, j), \quad T(i,j) \leq T(i, j+1)$ \\
(2) Each $i (i=1,2,3, \cdots m)$ appears at most once in each column.\\
(3) Each $i (i=1',2',3' , \cdots m')$ appears at most once in each row. \\
Let $ \alpha_k$=number of $T(i,j)=k $ or $k'$. Then for each tableau $T$ of a shifted Young diagram  we associate with a monomial  
\[ \eta (T)= a_1^{\alpha_1}a_2^{\alpha_2}a_3^{\alpha_3} \cdots a_m^{\alpha_m}. \]

Then we have the following relation \cite{st}
\be 
Q_{\lambda}(c)=\sum_T   \eta (T) ,    \label{st}                           
\ee 
where 
\be  c_{2k+1}=\frac{ a_1^{2k+1}+a_2^{2k+1}+a_3^{2k+1}+a_4^{2k+1} + \cdots+a_m^{2k+1} }{2k+1}
\label{rec}
\ee 
and the sum is over all the tableaux of the shifted Young diagram. For example, we take $m=4$ and then,  using  (\ref{qf})and (\ref{rec}),
\bean
q_{21}(c) &=& q_2(c) q_1(c) -2q_3(c) =4/3c_1^3-4c_3 \\
&=& (4a_1 + 4a_2 + 4a_4)a_3^2 + 4(a_1 + a_2 + a_4)^2a_3 + 4(a_2 + a_4)(a_1 + a_4)(a_1 + a_2) \\
q_{31}(c) &=& q_3(c) q_1(c) -2q_4(c) = 4c_1(c_1^3 - 3c_3)/3 \\
&=& 4(a_4 + a_3 + a_2 + a_1)[(a_4 + a_3 + a_2)a_1^2 + (a_4 + a_3 + a_2)^2a_1  \\
&+ & (a_4 + a_3)(a_2 + a_4)(a_3 + a_2)] \\
q_{32}(c) &=& q_3(c) q_2(c) -2q_4(c)q_1(c)+ 2q_5(c)= 8/15c_1^5 - 4c_1^2c_3 + 4c_5\\
&=& 4(a_1 + a_2 + a_4)^2a_3^3 \\
&+&  4[a_4^3 + 4(a_1 + a_2)a_4^2 + 4(a_1 + a_2)^2a_4 + a_1^3 + 4a_1^2a_2 + 4a_1a_2^2 + a_2^3]a_3^2 \\
&+& 8(a_2 + a_4)(a_1 + a_4)(a_1 + a_2)(a_1 + a_2 + a_4)a_3 \\
&+&  4[(a_1 + a_2)a_4 + a_1a_2](a_1 + a_2)(a_1 + a_4)(a_2 + a_4) \\
Q_{3210}(c)&=&  q_{32}(c)q_1(c)- q_{31}(c)q_2(c)+ q_3(c)q_{21}(c)\\
&=& 8/45c_1^6 - 8/3c_3c_1^3 - 8c_3^2 + 8c_1c_5\\
&=& 8(a_4 + a_3)(a_2 + a_4)(a_3 + a_2)(a_1 + a_4)(a_1 + a_3)(a_1 + a_2)
\eean
Then we have the following 
\begin{theorem}
Let $\lambda$ be a DP. Then $Q_{\lambda}(a) $ is a totally positive Pfaffian if $ a_i >0$ for $i=1,2,3, \cdots m$ defined in (\ref{rec}). 
\end{theorem}
This theorem can be generalized to skew Schur's Q function as follows. Given two strict Young diagrams  $ \mu$ and $\lambda$ with $ \mu \subseteq \lambda$, a shifted skew diagram is defined as $Q_{\lambda  / \mu}$ and its corresponding tableau is defined similarly above. Please see the figure 2. 
\begin{figure}[t]
	\centering
		\includegraphics[width=1\textwidth]{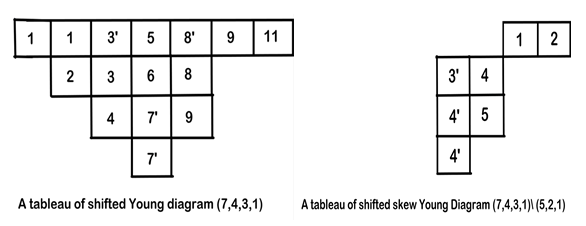}
	\caption{The monomial  of the left panel is $ \eta(T)=a_1^2 a_2a_3^2a_4 a_5 a_6 a_7^2a_8^2a_9^2 a_{11}$, and the monomial  of the right  panel is $ \eta (T)=a_1a_2a_3a_4^3 a_5$.} 
\end{figure}
We have the relation \cite{pj}
\be 
Q_{\lambda / \mu}(c)=\sum_T   \eta (T) ,    \label{sk}                           
\ee 
where $ Q_{\lambda / \mu}(c) $ is defined by 
\be 
Q_{\lambda / \mu}(c)=\left \{ \ba{cc} 
 Pf \left [\ba {cc} Q_{\lambda}(c) & M_{\lambda / \mu} (c)   \\ - M_{\lambda / \mu}^T  (c)  & O  \ea \right ]  \mbox{ if $l (\lambda)=l (\mu)  \qd mod \qd 2$ }\\
Pf \left [\ba{cc} Q_{\lambda^0}(c) & M_{\lambda^0 / \mu} (c)   \\ - M_{\lambda^0 / \mu}^T  (c) &  O    \ea \right ]
 \mbox{ if $l (\lambda) \neq l (\mu)  \qd mod \qd 2$ }
\ea \right.  \label{sub}
\ee
Here $l( \cdot )$ is  the length of  a strict Young diagram and  $\lambda^0 $ is $ (\lambda_1,  \lambda_2 ,  \lambda_3,  \cdots ,  \lambda_{n}>0 ,   0) $; moreover,  
\be 
 M_{\lambda / \mu} (c)=\left [Q_{(\lambda_i-\mu_{s+1-t}) }\right],  \qd 1 \leq t \leq s=l (\mu) , \label{bd}
\ee
where $\mu= ( \mu_1, \mu_2, \cdots , \mu_s) $ and $ Q_{k}=0 $ for $ k<0$. 
For example,  
\be
 Q_{(653/42)} (c) =Q_{(6530)/(42)}(c) = Pf \left [ \ba{cccccc} 0 & Q_{65} & Q_{63} & Q_6 & Q_4 & Q_2 \\
-Q_{65} &0 & Q_{53} & Q_{5} &Q_3 & Q_1 \\
-Q_{63} & -Q_{53} & 0 & Q_3 & Q_1 &0 \\
-Q_6 & -Q_{5} & -Q_3 &0 & 0&0 \\
-Q_4 & -Q_{3} & -Q_1 &0 &0 &0 \\
-Q_2 & -Q_{1} & 0 &0 &0 &0 \ea  \right], \label{ex}
\ee
where $ l(\lambda)=3$ , $ l(\mu)=2$ and 
\be M_{\lambda / \mu}(c)= \left [\ba {cc} Q_4 & Q_2 \\Q_3 & Q_1 \\ Q_1 &0 \\  0&0 \ea  \right]. \label{xx} \ee
 To consider totally-nonnegative Pfaffian, for any strict  Young diagram $\lambda'  \subset \lambda$ and  strict  Young diagram $ \mu' \subset \mu$, it's not difficult to  see that \cite{pj}
\[ Q_{\mu' \subset \lambda' } >0, \qd \mbox{and}  \qd  Q_{\mu' \nsubseteq \lambda' } =0 . \] 
As the last example (\ref{ex}), 
\[ Q_{(63)/(42)}(c) =Pf  \left [\ba {cccc} 0 &Q_{63}  &Q_4 & Q_2 \\
    -Q_{63} &0 &Q_1 & 0 \\   -Q_4 &-Q_{1}  &0 & 0   \\ -Q_2 &0  &0& 0
\ea \right]= Q_2 Q_1 >0. \]
Then we also have the following 
\begin{theorem}
The skew-symmetric matrix defined in (\ref{sub}) is a totally non-negative Pfaffian. 
\end{theorem}

\indent Finally, one remarks, similar to the Plucker relations in totally non-negative Grassmannian \cite{ko3, ko6}, that we also have the Pfaffian-Plucker  relations in totally non-negative Pfaffian \cite{ba, oh}. Let $p$ and $r$ be two odd numbers and 
$\{ i_1, i_2 , i_3, \cdots, i_p \} $ and $\{ j_1, j _2 , j_3, \cdots, j_r \} $ be any two subsets of  $\{ \lambda_1, \lambda_2, \cdots , \lambda_{2n} \} $. Then the Pfaffian-Plucker  relation is 
\bea 
&& \sum_{k=1}^p (-1)^{k-1} Pf ( i_1, i_2, i_3 , \cdots, \hat i_k,\cdots, i_p) Pf (i_k, j_1, j _2 , j_3, \cdots, j_r) \no \\
&& = \sum_{s=1}^r  (-1)^{s-1} Pf ( i_1, i_2, i_3 , \cdots, i_p, j_s) Pf (j_1, j _2 , j_3,  \cdots, \hat j_s , \cdots, j_r), \label{pp}
\eea
where  $\hat i_k$ means the $i_k$ term is deleted and similarly for $\hat j_s$.\\

\section{Totally Non-negative Pfaffian in Block Form }
In this section, the block form for the totally non-negative Pfaffian is studied and the Cauchy-Bitnet formula is utilized  to obtain the transformation to preserve the totally non-negativity. 

We begin with the following
\begin{theorem}\cite{ok} \\
Suppose that $n+m$ is even. If $Z$ is an $ n \times n $ skew-symmetric matrix and $W$ is an $ n \times m $ matrix, then we have 
\be
Pf \left [\ba {cc} Z& W   \\ - W^T    & O_{m,m}  \ea \right ]= 
\left \{ \ba{lll}    \sum_{I} (-1)^{\sum (I)+ {n \choose 2}} Pf [Z(I)] det W ([n] \bk  I; [m]) && \mbox{ if $  n>m $ }   \\
 (-1)^{n \choose 2 }  det W  && \mbox{ if $  m=n $ }  \\
 0     && \mbox{ if $  n<m  $ },              \label{ok}          \ea       \right. 
\ee
where $[n]=\{1,2,3, \cdots, n\}, [m]=\{1,2,3, \cdots, m\}$ and $I$ runs over all $(n-m)$-elements of $[n]$. Also, 
$Pf(\emptyset)=1$.
\end{theorem}
Then we have the following 
\begin{theorem}
Let $ l(\lambda)=n,  l(\mu)=m$ and $n+m$ is even. Then the  anti-block form 
\be 
\hat M_{\lambda / \mu}= \left [\ba {cc} O_{n, n} & M_{\lambda / \mu} (c)   \\ - M_{\lambda / \mu}^T  (c)  & O_{m, m}      \ea \right ], 
  \label{an}
\ee
where $ M_{\lambda / \mu} (c)$ is defined in (\ref{bd}),  is a totally non-negative Pfaffian. 
\end{theorem}
\begin{proof}
In the skew-symmetric matrix in (\ref{sub}), to obtain any non-zero principal minor, one chooses any first $k$ rows, $\{ i_1, i_2, \cdots, i_k \} \subset [n] $ and  another last  $k$ rows, $ \{ j_1, j_2, \cdots , j_k \}  \subset \{n+1, n+2, \cdots , n+m \}$. Then one chooses the columns \\
$ \{ i_1, i_2, \cdots, i_k, j_1, j_2, \cdots , j_k\} $ to form a $ 2k \times 2k$ skew-symmetric matrix, which has the block form $m=n=k$ in (\ref{ok}).  This Pfaffian is independent of $Z$ and we are able to set $Z=O$. This completes the proof. 
\end{proof}

We illustrate the proof by the example (6530)/(42), n=4, m=2. In the matrix (\ref{ex}), one chooses $ \{i_1=1, i_2=2\} $ and $\{ j_1=5, j_2=6 \} $ . Then the skew-symmetric matrix corresponding to (65)/(42) is 
\[ Q_{(65)/(42)}=\left [\ba {cccc} 0 &Q_{65}  &Q_4 & Q_2 \\
    -Q_{65} &0 &Q_3 & Q_1 \\   -Q_4 &-Q_{3}  &0 & 0   \\ -Q_2 &-Q_1  &0& 0
\ea \right], and  \qd \hat M_{(65)/(42)}=\left [\ba {cccc} 0 &0  &Q_4 & Q_2 \\
    0 &0 &Q_3 & Q_1 \\   -Q_4 &-Q_{3}  &0 & 0   \\ -Q_2 &-Q_1  &0& 0
\ea \right] .  \]
Then 
\[Pf (Q_{(65)/(42)})= Pf (\hat M_{(65)/(42)})= (-1) det W=(-1) [( Q_4 Q_1)-(Q_2 Q_3)]=Q_2 Q_3-Q_4 Q_1 >0, \]
where  
\[W= \left [\ba{cc} Q_4 & Q_2 \\ Q_3 & Q_1 \ea \right]. \]
The pfaffian is independent of $Q_{65}$. We notice that the pfaffian of $ \hat M_{\lambda / \mu}$ defined in (\ref{an}) is zero using the result (\ref{ok}) on the condition $ n > m$. \\

Next, we consider the transformations to preserve the property of being totally nonnegative pfaffian. Let's recall the Cauchy-Binet formula. Suppose $A$ and $B$ are $m \times n$ and $n \times p$ matrices respectively. Let $C=AB$. Assume  $ \alpha \subset[m],  \beta \subset[p]$ and $| \alpha |=| \beta |=k \leq min (m,n,p)$. Then we have the Cauchy-Binet formula between the minors 
\be det C( \alpha , \beta)= \sum_{\gamma} det A( \alpha , \gamma) det B(\gamma , \beta), \label{ba} \ee
where the sum is taken over the subsets of $n$ with $| \gamma|=k$.  Let 
\be 
\hat W= \left [\ba {cc} O_{n, n} & W  \\ - W^T  & O_{m, m}      \ea \right ], 
  \label{po}
\ee
where $n+m$ is even and $W$ is a $ n \times m$ matrix ( $ n \geq m$) . Then we have the 
\begin{theorem}
Suppose the skew-symmetric matrix (\ref{po}) is totally non-negative Pfaffian, and define 
\be 
\hat {W}_L= \left [\ba {cc} O_{n, n} & LW  \\ - (LW)^T    & O_{m, m}      \ea \right ]  \qd and \qd 
\hat {W}^R= \left [\ba {cc} O_{n, n} & WR  \\ - (WR)^T    & O_{m, m},      \ea \right ] \label{ma}
\ee
where $ L$ and $R$ are $n \times n$ and $ m \times m $ totally non-negative matrices, respectively. Then both $ \hat {W}_L$ and $ \hat {W}^R$ are totally non-negative Pfaffians.
\end{theorem}
This theorem is a result of  the Cauchy-Binet formula (\ref{ba}) and the case $m=n$ in (\ref{ok}). 

Let's take the example (\ref{ex})  and consider the totally non-negative matrix in the (0,1)-double echelon form 
\cite{fj}
\[
L= \left [\ba {cccc} 1 & 1 & 0 & 0   \\  0 & 1& 0& 0 \\ 0&1 &1 &1 \\ 0&0 &1 &1  \ea \right ] 
 \qd \mbox{and}   \qd    R= \left [ \ba {cc} 1&1 \\ 0& 1\ea \right]                                    \] 
Then 
\[ LM_{\lambda / \mu}= \left [\ba {cc} Q_3+Q_4 & Q_1+Q_2\\ Q_3 & Q_1\\ Q_1+Q_3& Q_1 \\  Q_1&0  \ea  \right].\qd \mbox{and}   \qd    M_{\lambda / \mu} R= \left [ \ba {cc} Q_4& Q_2+Q_4 \\ Q_3 &Q_1+Q_3 \\  Q_1 & Q_1 \\ 0 &0  \ea \right]        \]
One can verify directly both $\hat {W}_L$ and $\hat {W}^R$ are  totally non-negative Pfaffians. 

In particular, we assume $W$  has the anti-diagonal form
\be W=antdiag[1,1,1, \cdots, 1] =W^{-1} \label{ti}. \ee
Then we have 
\begin{theorem}
Given any $ n \times n$ totally non-negative matrix $L$ and any $ m \times m $ totally non-negative matrix $R$,  and assume $W$ has the form  (\ref{ti}), then both $\hat {W}_L$ and $\hat{W}^R$ in (\ref{ma}) are totally non-negative Pfaffians.
\end{theorem}
\begin{proof}
It is known that for any sub-matrix $W^{(k)}$ of size k in $W$ 
\[ det W^{(k)}= (-1)^{k \choose 2}.\] 
Also, using  (\ref{ok}), for any non-zero principal minor, one has 
\bean Pf ( {\hat W}_L(\alpha, \alpha) ) &=& (-1)^{k \choose 2} det LW^{(k)}(\alpha_1 , \alpha_2  \mbox{ mod  n} ) \\
&=& (-1)^{k \choose 2}(-1)^{k \choose 2} det L( \alpha_1 , \alpha_2   \mbox{ mod  n} )=det L(\alpha_1 , 
\alpha_2  \mbox{ mod  n} ),
\eean
where $\alpha=\{ \alpha_1, \alpha_2, \cdots, \alpha_k, \alpha_{k+1}, \alpha_{k+2},  \cdots , \alpha_{2k}\}  \subset  [2n],   \alpha_1 =\{ \alpha_1, \alpha_2, \cdots, \alpha_k \} \subset [n], $  and $  \alpha_2=\{ \alpha_{k+1}, \alpha_{k+2},  \cdots , \alpha_{2k}\} \subset \{ n+1, n+2, \cdots, n+m\} $.  Using the total non-negativeness of the matrix $L$, one obtains  $\hat {W}_L$ is totally non-negative pfaffian. A similar consideration obtains that $\hat{W}^R$ is  totally non-negative pfaffian.  
\end{proof}
\begin{theorem}
Let 
\[ \Omega= \left [\ba {cc}  Z & O  \\ O     & Z' \ea \right].  \]
If both $Z$ and $Z'$ are totally non-negative Pfaffians, then $\Omega$ is a totally non-negative Pfaffian. 
\end{theorem}
\begin{proof}
Use
\[  Pf   \left [\ba {cc}  Z^{(k)} & O  \\ O     & Z'^{(k')}  \ea \right ]   = Pf (Z^{(k)}) Pf (Z'^{(k')}),     \]
where $Z^{(k)}$ and $ Z'^{(k')}$ are sub-matrices of $Z$ and $Z'$ respectively with size $k$ and $k'$. 
\end{proof}

On the other hand, to construct  a totally non-negative matrix,  using (\ref{ti}) and $M_{\lambda / \mu} $ in (\ref{an}), one defines 
\be M_L=W M_{\lambda / \mu} \qd  \mbox{and} \qd M_R= M_{\lambda / \mu}W. \label{mm} \ee
Here the first $W$ has  size  $ l (\lambda) $ and the second $W$ has size  $  l (\mu) $. We also have the 
\begin{theorem}
Both $M_L$ and $M_R$ defines in (\ref{mm}) are totally non-negative matrix. 
\end{theorem}
\begin{proof}
 In Theorem 3.2, it's known that for any index $K= \{ i_1, i_2, \cdots, i_k \} \subset [l (\lambda)] $, the sub-matrix  $M_{\lambda / \mu} (K)$  of size k obtained from the $ n \times m$ matrix  $M_{\lambda / \mu} $ has the determinantal relation 
\[ (-1)^{ k \choose 2} det M_{\lambda / \mu} (K) >0.\]
Also, from (\ref{mm}), we have 
\bean 
(-1)^{ k \choose 2} det M_{\lambda / \mu} (K) &=& (-1)^{ k \choose 2} det W^{-1(k)} M_L (K)= (-1)^{ k \choose 2}(-1)^{ k \choose 2} det M_L (K)  \no \\
&=&  det M_L (K)  >0.
\eean 
It proves $M_L$ is a totally non-negative matrix. For the case $M_R$, the proof is similar.
\end{proof} 
From this theorem, using skew Schur Q function, we can construct totally non-negative matrix in $(a_1, a_2, a_3, \cdots , a_m)$. For example, letting $\lambda=(9654) $ and $\mu=(321)$, we have 
\[ M_{\lambda / \mu}=  \left [\ba {ccc} Q_8 & Q_7  &  Q_6 \\ Q_5 & Q_4  &  Q_3 \\     
Q_4 & Q_3  &  Q_2 \\   Q_3 & Q_2  &  Q_1    \ea \right ].  \]
Then 
\[ M_L=WM_{\lambda / \mu}= \left [\ba {ccc}   Q_3 & Q_2  &  Q_1 \\ Q_4 & Q_3  &  Q_2 \\ Q_5 & Q_4  &  Q_3 \\ Q_8 & Q_7  &  Q_6  \ea \right ],  \]
and 
\[ M_R=M_{\lambda / \mu} W =  \left [\ba {ccc} Q_6 & Q_7  &  Q_8 \\ Q_3 & Q_4  &  Q_5 \\     
Q_2 & Q_3  &  Q_4 \\   Q_1 & Q_2  &  Q_3    \ea \right ].  \]
One can verify directly that both $M_L$ and $M_R$ are totally non-negative matrices by software. 

The $ \tau$- function define in (\ref{tc}) is simplified if one considers the totally non-negative Pfaffian in the block form (\ref{po}) and the structure of corresponding resonance of web solitons is still unknown. 

\section{ Totally non-negative  Pfaffian in factorization}
In this section, we consider the singular anti-symmetric matrix in factorization , that is, the pfaffian is zero and it has the form \cite{yt}
\be
A=S^T J S, \label{si}
\ee
where $S$ is an $ 2r \times m ( m > 2r) $ matrix and $J$ is an anti-symmetric $ 2r \times 2r$ matrix. Then $A$ is an anti-symmetric $m \times m$ matrix with rank $ \leq 2r$.  Notice that  $m$ could be odd and then $Pf(A)=0$.  
The factorization form (\ref{si}) is connected with the totally non-negative Grassmannian  \cite{ko6}.
\begin{theorem}\cite{is}\\
$ Pf(A)= \sum_{I \subset [m], |I|=2r} Pf ( J_I)  det (S_{ I}). $
 \end{theorem}
From this theorem, we see that if one considers any principal minor of $A$, then 
\be 
Pf(A_{\alpha})= \sum_{K  \subset [2r], |K|=|\alpha|} Pf ( J_K)  det (S_{ K, \alpha}),  \label{su}
\ee 
where $ |\alpha| \leq m$ and $ (S_{ K, \alpha}) $ is the matrix obtained from S by choosing the  K-rows corresponding to the index K and the $\alpha$-columns corresponding to the index $\alpha$. 

Now, we take 
\be J= antdiag[1,1,1, \cdots, 1, -1, -1, -1 , \cdots , -1]. \label{js} \ee 
There are r 1s and r (-1)s. Any principal minor of J is 0 or 1.  By (\ref{su}), one has 
\be
Pf (A_{\alpha})= 
\left \{ \ba{ll}    0 & \mbox{ if $  |\alpha| >2r $ }   \\
 \sum_{\{k_1, k_2, \cdots , k_{ |\alpha|/2} \}  \subset [r], \{ k_{ \frac{|\alpha|}{2}+1 }, \cdots , k_{|\alpha|}\}  \subset [r+1, 2r]} det (S_{ K, \alpha})   & \mbox{ if $ |\alpha| \leq 2r$ },  \label{kk}          \ea       \right. 
\ee
where $K=\{k_1, k_2, k_3, \cdots , k_{|\alpha|}\}$. The summation is over all the principal minors of J being equal to 1 , and  there are $ {r \choose |\alpha|/2}^2$  terms. \\

On the other hand, if  J has the diagonal block form 
\be  J= diag[   \left[ \ba{ll}    0 & 1   \\ -1&0   \ea       \right],   \left[ \ba{ll}    0 & 1   \\ -1&0   \ea       \right], \cdots ,        \left[ \ba{ll}    0 & 1   \\ -1&0   \ea       \right]  ],
 \label{bo} \ee
then a formula similar to (\ref{kk}) is
\be
Pf (A_{\alpha})= 
\left \{ \ba{ll}    0 & \mbox{ if $  |\alpha| >2r $ }   \\
 \sum_{K \subset \{\{1,2\}, \{3,4\}, \{5,6\}, \cdots, \{2r-1, 2r\} \},  |K|=|\alpha| }det (S_{ K, \alpha})   & \mbox{ if $ |\alpha| \leq 2r$ }.  \label{ko}          \ea       \right. 
\ee
Here there are ${r \choose |\alpha|/2}$ terms in this summation. 
\begin{theorem}
Let $J$ be defined in (\ref{js})  ( or \ref{bo}) . If $S$ is an $ 2r \times m ( m > 2r) $ matrix such that the  corresponding all even minors  in (\ref{kk}) ( or \ref{ko}) are  totally non-negative, then the matrix $A$ defined in (\ref{si}) is a totally non-negative Pfaffian. 
 \end{theorem}
The even minors here mean $ |K|(=|\alpha| ) $ is even. In \cite{ko3, ko6}, the totally non-negative Grassmannian (TNG) is used to describe the resonant theory of multi-line solitons of KP equation. Given an $ 2r \times m ( m > 2r) $ matrix, it is a TNG if and only if its each sub-determinant of $ 2r \times 2r$ matrix is non-negative. The condition of  some non-negative even minors is stronger than TNG but weaker than totally non-negative matrix. 

Let's consider $r=1$ and $m \geq 3$. Suppose  
\be 
S=  \left[ \ba{llllllll}   s_{11}    &  s_{12} & s_{13} & s_{14} & s_{15} & s_{16 } & \cdots &  s_{1m}\\
  s_{21}  & s_{22}  & s_{23} & s_{24} & s_{25} & s_{26 } & \cdots &  s_{2m}   \ea       \right] , \qd   
	J= \left[ \ba{ll}    0 & 1 \\ -1 & 0   \ea       \right],      \label{to} 
\ee
 To obtain TNG, the sub-determinants of $S$ will be 
\be  \Delta_{ij}= det \left[ \ba{ll}   s_{ 1i}  &  s_{ 1j} \\   s_{ 2i }&  s_{ 2j} \ea       \right]= s_{ 1i}s_{ 2j}-s_{ 1j}s_{ 2i}= - \Delta_{ji} \geq 0  \label{sb} \ee
with the Plucker's relations
\be \Delta_{ik}\Delta_{jl}= \Delta_{ij}\Delta_{kl}+ \Delta_{il}\Delta_{jk}, \label{pu} \ee
where $ 1 \leq i \leq j \leq m$ .  
We remark that the conditions (\ref{sb}) with constraints (\ref{pu}) are under-determined and  their  solutions can be described  by the Le-diagrams defined on Young diagrams \cite{po} ( also  see \cite{ko6}).

A simple calculation obtains  that the anti-symmetric $ m \times m $ matrix A defined in  (\ref{si}) is
\be  A=[\Delta_{ij}]. \label{td} \ee
From (\ref{kk}) and (\ref{td}), the $\tau$-function for BKP is 
\be 
\tau_A =1+ \sum_{1 \leq i<  j \leq m} \frac{(p_i-p_j)}{2(p_i+p_j)} \Delta_{ij} e^{x(p_i+p_j) +y (p_i^3+p_j^3) + t( p_i^5+p_j^5 )},
 \label{Ta} 
\ee
where $ p_1^2  > p_2^2 > p_3^2> \cdots > p_{2n}^2$. 
For example,  one considers the TNG of  T-type soliton in KP theory \cite{ko6}
\be 
S=  \left[ \ba{llll}   1   &  0 & -a & -b   \\
  0  & 1 & c &  d  \ea       \right],   \label{td}
\ee
where $ a,b c, d$ are positive numbers and $ bc-ad >0$. Then 
\be A=  \left[ \ba{llll}   0  &  1 & c & d   \\
  -1  & 0 & a &  b  \\
	-c & -a & 0& bc-ad \\
	-d & -b & ad-bc &0  \ea       \right]= \left[ \ba{llll}   0  &  A_{12} & A_{13} &   A_{14} \\
  -A_{12}  & 0 & A_{23} &  A_{24}  \\
	-A_{13} &-A_{23} & 0& A_{34} \\
	-A_{14} & -A_{24} & -A_{34} &0  \ea       \right].   \label{tt}  \ee
Notice that $Pf(A)=0$. Then the $\tau$-function is 
\bea 
\tau_A &=&  1+ \frac{p_1-p_2}{2 (p_1+p_2)} e^{x(p_1+p_2) +y (p_1^3+p_1^3) + t( p_1^5+p_2^5 )}+A_{13} \frac{p_1-p_3}{2 (p_1+p_2)} e^{x(p_1+p_3) +y (p_1^3+p_3^3) + t( p_1^5+p_3^5 )} \no \\
&+& A_{14} \frac{p_1-p_4}{2 (p_1+p_4)} e^{x(p_1+p_4) +y (p_1^3+p_4^3) + t( p_1^5+p_4^5 )}+ A_{23}\frac{p_2-p_3}{2 (p_2+p_3)} e^{x(p_2+p_3) +y (p_2^3+p_3^3) + t( p_2^5+p_3^5 )} \no \\
&+ &  A_{24}\frac{p_2-p_4}{2 (p_2+p_4)} e^{x(p_2+p_4) +y (p_2^3+p_4^3)t( p_2^5+p_4^5 )}  +   A_{34}   \frac{p_3-p_4}{2 (p_3+p_4)} e^{x(p_3+p_4) +y (p_3^3+p_4^3) + t( p_3^5+p_4^5 )},   \no \\
\label{tu}
\eea 
where $ p_1^2  > p_2^2 > p_3^2> p_{4}^2$. Please see the figure 3.

\section{Concluding Remarks}

The $\tau$-function of BKP has a Pfaffian structure and the coefficients can be expressed by skew Schur's Q functions. To obtain non-singular solitons, one investigates the totally non-negative Pfaffian and studies two special types : block type  and factorization type. Using Cauchy-Bitnet formula, the  totally non-negative Pfaffian is preserved under the multiplication of a totally non-negative matrix. In \cite{ko3, ko6}, the resonant structure of KP equation is investigated by the Le-Diagram \cite{po} using the totally non-negative Grassmannian. It's known that the totally non-negative matrix can also be described the Le-Diagram \cite{po}. We hope the resonant structure of BKP could be studied in a similar way when  $ t \to \pm \infty $. Furthermore, given a  totally non-negative Grassmannian \cite{ko6}, one could possibly construct a totally non-negative Pfaffian using (\ref{kk}) or (\ref{ko}). These need further investigations.
\subsection*{Acknowledgments}The author thanks Prof. M. Ken-ichi for his introducing the reference \cite{yt} and fruitful discussions. This work is supported in part by the National Science and Technology Council of Taiwan under
Grant No. NSC 113-2115-M-606-001.

\begin{figure}[t]
	\centering
		\includegraphics[width=0.88\textwidth]{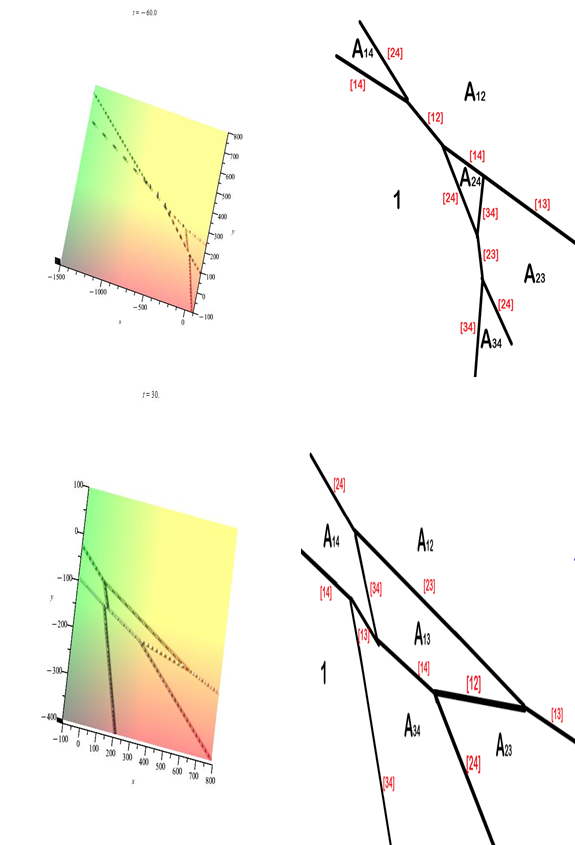}
	\caption{ $p_1=2, p_2=1.5, p_3=0.8, p_4=0.1, a=5, b=7, c=3, d=3. $ The soliton graphs  correspond to $ t <<0$ (top panels) and $ t>> 0$ (bottom panels).  When $ |y| \to \infty$,  the unbounded line solitons are invariant and we notice the  triangle and quadrilateral due to the resonances.  The quadrilateral also appears for the T-type soliton in KP theory; however, there is no such a triangle in KP theory. }
\end{figure}

\end{document}